\documentclass[12pt,a4paper]{article}

\usepackage{amsmath}
\usepackage{amssymb}
\usepackage{graphicx}
\usepackage{amsthm}
\usepackage{amsfonts}
\usepackage{amscd}
\usepackage{braket}
\usepackage{bm}
\usepackage{cite}
\usepackage{url}

\theoremstyle{plain}
\newtheorem{thm}{Theorem}[section]

\newtheorem{theorem}[thm]{Theorem}

\newtheorem{lemma}[thm]{Lemma}

\newtheorem{definition}[thm]{Definition}

\newtheorem{remark}[thm]{Remark}
\newtheorem{fact}[thm]{Fact}

\begin{document}
\title{
A Four-Qubits Code that is a Quantum Deletion Error-Correcting Code
 with the Optimal Length}

\author{
Manabu HAGIWARA
\thanks{
Department of Mathematics and Informatics,
Graduate School of Science,
Chiba University
1-33 Yayoi-cho, Inage-ku, Chiba City,
Chiba Pref., JAPAN, 263-0022
}
 \and
Ayumu NAKAYAMA
\thanks{
Department of Mathematics and Informatics,
Graduate School of Science and Engineering,
Chiba University
1-33 Yayoi-cho, Inage-ku, Chiba City,
Chiba Pref., JAPAN, 263-0022
}
}

\date{}
\maketitle

\begin{abstract}
This paper provides a new instance of quantum deletion error-correcting codes.
This code can correct any single quantum deletion error,
while our code is only of length 4.
This paper also provides an example of an encoding quantum circuit and decoding quantum circuits.
It is also proven that the length of any single deletion error-correcting codes is greater than
or equal to 4.
In other words, our code is optimal for the code length.
\end{abstract}

\section{Introduction}
Similar to classical error-correcting codes,
quantum error-correcting codes are an essential factor for implementing
practical quantum communication and quantum computation.

In recent classical coding theory,
deletion codes have attracted researcher's attention.
For example,
while only three papers on deletion codes were presented at the symposium ISIT 2015,
more than ten papers were presented at ISIT 2017, 2018 and 2019.
Furthermore, two technical sessions for deletion codes were organized at the last ISIT.
Classical deletion error-correcting codes have applications to
reliable communication for synchronization error \cite{sala2017exact,helberg1993coding},
error-correction for DNA storage \cite{buschmann2013levenshtein},
error-correction for racetrack memory \cite{chee2017coding},
etc.

On the other hand,
only a few studies on quantum deletion correcting codes 
have been published.
In 2019, Leahy et.al. introduced quantum insertion/deletion channel
\cite{leahy2019quantum}
and a technique to reduce quantum deletion error to
quantum erasure error under the specific assumption.
The first quantum binary deletion codes under the general scenario
was constructed in \cite{nakayama2019first}.
The code length was $8$.

Deletion error-correction is a problem that
to determine the original information from a partial information.
The difference from erasure error-correction is that
we are not given the information on the error positions.
For example, a bit-sequence $00111000$ is changed to
$001?1000$ by a single erasure error
but
is changed to $0011000$ by a single deletion error.
The symbol ``$?$'' tells the position where error occurred.
An erasure bit-sequence $001?1000$ can be transformed to $0011000$ by deleting the symbol ``$?$''.
Hence deletion error-correction is more difficult than erasure error-correction.

In quantum information theory,
quantum deletion error-correction is a problem
to determine the quantum state in the entire quantum system from 
a quantum state in a partial system.
Therefore it is related to various topics, e.g.,
quantum erasure error-correcting codes \cite{grassl1997codes},
partial trace, quantum secret sharing \cite{cleve1999share,hillery1999quantum},
purification of quantum state \cite{hughston1993complete},
quantum cloud computing \cite{biamonte2017quantum}
and etc.

This paper provides a new and a shorter length single deletion error-correcting
code.
The code space is not an instance of previously known quantum error-correcting
codes, e.g.,
CSS codes \cite{calderbank1996good,steane1996multiple},
stabilizer codes \cite{gottesman1997stabilizer},
surface codes \cite{fowler2012surface},
and etc.
Remark that our code length is only $4$.
This is the same length to the optimal shortest code length of
quantum erasure error-correcting codes \cite{grassl1997codes}.
In fact, we show that
$4$ is also the optimal length of quantum deletion error-correcting codes.

This paper also provides examples of one encoding circuit and two different decoding circuits.
The depth of the circuits are small.

The readers are assumed to be familiar with quantum information theory
and coding theory, in particular,
quantum error-correcting codes and classical deletion error-correcting codes.

\section{Deletion Errors and Code Construction}
\subsection{Deletion Error and Deletion Error-Correcting Codes}

Set $\ket{0}, \ket{1} \in \mathbb{C}^2$ as 
$$\ket{0}:= \left( \begin{array}{c} 1 \\ 0 \end{array} \right),
\ket{1}:= \left( \begin{array}{c} 0 \\ 1 \end{array} \right) $$
respectively.
For a binary sequence $ \bm{x} = x_1 x_2 \dots x_n \in \{0,1\}^n$,
$ \ket{ \bm{x} } $ denotes
$$
\ket{ x_1 } \otimes \ket{ x_2 } \otimes \dots \ket{ x_n }
\in \mathbb{C}^{2 \otimes n} .
$$
We denote the set of all density matrices of order $N$
by $S(\mathbb{C}^N)$.
An element of $S(\mathbb{C}^N)$ is called a quantum state in this paper.
We also use a complex vector for representing a quantum state if 
the state is pure.

For an integer $1 \le i \le n$ and a square matrix
$$A = \sum_{\bm{x},\bm{y} \in \{0,1\}^n }
  a_{\bm{x},\bm{y}} \cdot \ket{x_1} \bra{y_1} \otimes \cdots \otimes \ket{x_n} \bra{y_n}$$
with $a_{\bm{x},\bm{y}} \in \mathbb{C}$,
 define the map $\mathrm{Tr}_i : S( \mathbb{C}^{2 \otimes n} ) \rightarrow S( \mathbb{C}^{2 \otimes (n-1)} )$ as
\begin{align*}
\mathrm{Tr}_i(A) :=
 &\sum_{\bm{x},\bm{y} \in \{0,1\}^n}
   a_{\bm{x},\bm{y}} \cdot \mathrm{Tr}(\ket{x_i} \bra{y_i}) \cdot
  \ket{x_1}\bra{y_1} \otimes \\
 & \cdots \otimes \ket{x_{i-1}} \bra{y_{i-1}}
  \otimes \ket{x_{i+1}} \bra{y_{i+1}} \otimes \\
 & \cdots \otimes \ket{x_n} \bra{y_n}.
\end{align*}
The map $\mathrm{Tr}_i$ is called a partial trace.

Recall that in the classical coding theory, a single deletion error is defined
as an operator that maps a sequence $x_1 x_2 \dots x_i \dots x_n$
to a short sequence $x_1 x_2 \dots x_{i-1} x_{i+1} \dots x_n$ for some $i$.

\begin{definition}[Deletion Error $D_i$]
For an integer $1 \le i \le n$,
we call $\mathrm{Tr}_i$ a single deletion error $D_i$,
i.e.,
$$D_i(\rho):= \mathrm{Tr}_i(\rho),$$
where $\rho \in S(\mathbb{C}^{2 \otimes n})$ is a quantum state.
\end{definition}
If a quantum state $\rho$ is corresponding to $n$-photons $p_1, p_2, \dots, p_n$,
the state $D_i( \rho )$ is corresponding to $n-1$-photons $p_1, p_2, \dots, p_{i-1}, p_{i+1}, \dots, p_n$.

\begin{definition}[Deletion Error-Correcting Code]
We call a vector space $Q \subset \mathbb{C}^{2 \otimes n}$
an $[n , k]$ single deletion error-correcting code if
\begin{itemize}
\item there exists a complex linear bijection
$\mathrm{Enc} : \mathbb{C}^{2 \otimes k} \rightarrow Q$,
\item there exists a map
$\mathrm{Dec} : S( \mathbb{C}^{2 \otimes (n-1)} ) \rightarrow \mathbb{C}^{2 \otimes k}$
such that
for any $ \ket{\phi} \in \mathbb{C}^{2 \otimes k}$ and for any $1 \le i \le n$,
$$
\mathrm{Dec} \circ D_i \circ \mathrm{Enc} ( \ket{ \phi })
=
\ket{ \phi },
$$
\end{itemize}
where $\circ$ is the composite for operations.
In other words,
there exist an encoder $\mathrm{Enc}$ and a decoder $\mathrm{Dec}$
that correct any single deletion errors.
\end{definition}

Comparing to erasure errors,
deletion errors do not tell the position where the information is deleted.
Hence to correct deletion errors is more difficult than to correct erasure errors.

\subsection{Code Construction}
Most famous classical error-correcting codes for single deletion errors are
``non-linear'' codes that are called VT codes \cite{tenengolts1965correction}
discovered by Levenshtein \cite{levenshtein1966binary}.
One of the reasons why ``non-linear'' codes are preferable for deletion error-correction
is unveiled in \cite{abdel2010correcting}. 
It is stated that a code rate of any single deletion error-correcting code
cannot exceed $1/2$ if a code is linear.
This implies that if a CSS code is constructed from two classical linear deletion error-correcting codes
then the code rate of the CSS code is $0$.

\begin{definition}[$\mathrm{En}_4$ and $Q_4$]
Let us define a linear map $\mathrm{En}_4 : \mathbb{C}^{2} \rightarrow \mathbb{C}^{2 \otimes 4}$.
For a quantum state $| \phi \rangle = \alpha | 0 \rangle + \beta | 1 \rangle \in \mathbb{C}^2$,
$\mathrm{En}_4$ maps the state $\ket{ \phi }$ to the following state $\ket{ \Phi }$,
\begin{align*}
| \Phi \rangle
:=&  \frac{\alpha}{\sqrt{2}} (| 0000 \rangle + | 1111 \rangle )\\
& + \frac{\beta}{\sqrt{6}}   (| 0011 \rangle + |0101 \rangle + |0110 \rangle \\
& + |1001 \rangle + |1010 \rangle + |1100 \rangle).
\end{align*}
Set $Q_4$ as the image of $\mathrm{En}_4$ for quantum messages, i.e.,
$$
Q_4 := \{ \mathrm{En}_4 ( \ket{ \phi } ) 
 \mid \ket{\phi} \in \mathbb{C}^2,
 \ket{\phi} \bra{\phi} \in S( \mathbb{C}^2 ) \}.
$$
\end{definition}
\begin{remark}
The QEC code with $4$ qubits \cite{grassl1997codes} is closely related to our code $Q_4$.
We can say that more symmetry is introduced to our code
for correcting error at the unknown position.
\end{remark}

In fact, we will show that 
$Q_4$ is a $[4,1]$ single deletion error-correcting code with the encoder $\mathrm{En}_4$.

We can characterize this encoding in the following manner.
Let $A$ be the set of four bit sequences with Hamming weight $0$ or $4$
and $B$ the set of four bit sequences with Hamming weight $2$,
i.e.,
\begin{align*}
A &= \{0000, 1111\},\\
B &= \{0011, 0101, 0110, 1001, 1010, 1100\}.
\end{align*}
Then the codeword $| \Phi \rangle$ is
$$
\frac{\alpha}{\sqrt{2}}  \sum_{a \in A} | a \rangle
+ \frac{\beta}{\sqrt{6}} \sum_{b \in B} | b \rangle.
$$

By this encoding, $\ket{0}$ and $\ket{1}$ are encoded to quantum states
 which are superpositions of two and six orthonormal states.
Hence this encoding is neither an instance of CSS codes nor stabilizer codes.

\begin{definition}[$\mathrm{De}_4$]
Let us define a map $\mathrm{De}_4$ from $S( \mathbb{C}^{2 \otimes 3} )$ to
$\mathbb{C}^{2}$.
The map $\mathrm{De}_4$ consists of the following steps:
\begin{itemize}
\item[] (Step 1) Perform the measurement $\{ P_0, P_1 \}$ to a quantum state in $S( \mathbb{C}^{2 \otimes 3} )$
and obtain the outcome $i \in \{0, 1\}$,
where $P_0$ (resp. $P_1$) is the projection from $\mathbb{C}^{2 \otimes 3}$ to
the linear space $V_0$ (resp. $V_1$) spanned by $\ket{000}, \ket{011}, \ket{101}$ and $\ket{110}$
(resp. $\ket{111}, \ket{100}, \ket{010}$ and $\ket{001}$).
Hence, the outcome is $i$ if the quantum state in $S( \mathbb{C}^{2 \otimes 3})$
 is changed into a state in $S( V_{i} )$.
\item[] (Step 2) Only if the outcome is $1$, act the quantum operation $F : V_1 \rightarrow V_0$ to $S$,
where
$F$ is defined as
\begin{align*}
F( a \ket{111} + b \ket{100} + c \ket{010} + d \ket{001} ) \\
:= a \ket{000} + b \ket{011} + c \ket{101} + d \ket{110}.
\end{align*}
\item[] (Step 3) Act the quantum operation $G : V_0 \rightarrow \mathbb{C}^{2 \otimes 3}$ to the state,
where
$G$ is defined as
\begin{align*}
G( a \ket{000} + b \ket{\overline{100}} + c \ket{\overline{010}} + d \ket{\overline{001}}\\
:= a \ket{000} + b \ket{100} + c \ket{010} + d \ket{011},
\end{align*}
and where
\begin{align*}
\ket{\overline{100}} &= \frac{1}{\sqrt{3}} \left( \ket{011} + \ket{101} + \ket{110} \right)\\
\ket{\overline{010}} &= \frac{1}{\sqrt{3}} \left( \ket{011} + \omega \ket{101}   + \omega^2 \ket{110} \right)\\
\ket{\overline{001}} &= \frac{1}{\sqrt{3}} \left( \ket{011} + \omega^2 \ket{101} + \omega   \ket{110} \right).
\end{align*}
\item[] (Step 4) Finally,
delete the 3rd and the 2nd qubits
by the partial traces $\mathrm{Tr}_3$ and $\mathrm{Tr}_2$.
\end{itemize}
\end{definition}

\subsection{Proof for Error-Correcting Property}
In this subsection, we prove our code $Q_4$
corrects any single deletion errors $D_1, D_2, D_3,$ and $D_4$.

Let us set 
\begin{align*}
A_0 &:= \{000\},\\
A_1 &:= \{111\},\\
B_0 &:= \{011, 101, 110\},\\
B_1 &:= \{001, 010, 100\}.
\end{align*}

\begin{lemma}
Let 
$\ket{\Phi} := \alpha \sum_{ \bm{a} \in A } | \bm{a} \rangle
+ \beta \sum_{ \bm{b} \in B} | \bm{b} \rangle$
and
$\rho := | \Phi \rangle \langle \Phi |$.
For any $1 \le i \le 4$,
$$D_i (\rho)
= \frac{1}{2}| \Phi_0 \rangle \langle \Phi_0 | +  \frac{1}{2} | \Phi_1 \rangle \langle \Phi_1 |,
$$
where
$| \Phi_0 \rangle = 
\alpha \sum_{ \bm{a} \in A_0 } | \bm{a} \rangle
+ \frac{\beta}{\sqrt{3}} \sum_{\bm{b} \in B_0} | \bm{b} \rangle$
and
$| \Phi_1 \rangle = 
\alpha \sum_{ \bm{a} \in A_1 } | \bm{a} \rangle
+ \frac{\beta}{\sqrt{3}} \sum_{\bm{b} \in B_1} | \bm{b} \rangle.$
\end{lemma}
\begin{proof}
At the first, we show a case where $i=1$.

By using $A_0, A_1, B_0$ and $B_1$,
we can rewrite $\ket{ \Phi }$ as
\begin{align*}
\ket{ \Phi }
=
&\ket{ 0 }
 \left(
   \frac{ \alpha }{\sqrt{2}} \sum_{\bm{a} \in A_0 } \ket{ \bm{a} }
    +
   \frac{ \beta }{\sqrt{6}} \sum_{\bm{b} \in B_0 } \ket{ \bm{b} }
 \right)\\
&+\ket{ 1 }
 \left(
   \frac{ \alpha }{\sqrt{2}} \sum_{\bm{a} \in A_1 } \ket{ \bm{a} }
    +
   \frac{ \beta }{\sqrt{6}} \sum_{\bm{b} \in B_1 } \ket{ \bm{b} }
 \right).
\end{align*}
Hence,
\begin{align*}
\rho
 &= \ket{ \Phi} \bra{ \Phi }\\
 &= \ket{ 0} \bra{ 0 } \otimes \\ 
 &  \biggl(
     \frac{ \alpha }{\sqrt{2}} \frac{ \overline{\alpha} }{\sqrt{2}}
        \sum_{\bm{a }, \bm{ a'} \in A_0 } \ket{ \bm{a} } \bra{ \bm{a '} }
  +
     \frac{ \alpha }{\sqrt{2}} \frac{ \overline{\beta} }{\sqrt{6}}
        \sum_{\bm{a } \in A_0, \bm{b' } \in B_0 } \ket{ \bm{a} } \bra{ \bm{b'} }\\
 &+
     \frac{ \beta }{\sqrt{6}} \frac{ \overline{\alpha} }{\sqrt{2}}
        \sum_{\bm{b } \in B_0, \bm{ a'} \in A_0 } \ket{ \bm{b} } \bra{ \bm{a '} }
  +
     \frac{ \beta }{\sqrt{6}} \frac{ \overline{\beta} }{\sqrt{6}}
        \sum_{\bm{b }, \bm{b' } \in B_0 } \ket{ \bm{b} } \bra{ \bm{b'} }
   \biggr)\\
 &+ \ket{ 1} \bra{ 1 } \otimes \\ 
 &  \biggl(
     \frac{ \alpha }{\sqrt{2}} \frac{ \overline{\alpha} }{\sqrt{2}}
        \sum_{\bm{a }, \bm{ a'} \in A_1 } \ket{ \bm{a} } \bra{ \bm{a '} }
  +
     \frac{ \alpha }{\sqrt{2}} \frac{ \overline{\beta} }{\sqrt{6}}
        \sum_{\bm{a } \in A_1, \bm{b' } \in B_1 } \ket{ \bm{a} } \bra{ \bm{b'} }\\
 &+
     \frac{ \beta }{\sqrt{6}} \frac{ \overline{\alpha} }{\sqrt{2}}
        \sum_{\bm{b } \in B_1, \bm{ a'} \in A_1 } \ket{ \bm{b} } \bra{ \bm{a '} }
  +
     \frac{ \beta }{\sqrt{6}} \frac{ \overline{\beta} }{\sqrt{6}}
        \sum_{\bm{b }, \bm{b' } \in B_1 } \ket{ \bm{b} } \bra{ \bm{b'} }
   \biggr)\\
 &+ \ket{ 0} \bra{ 1 } \otimes \rho' + \ket{1} \bra{0} \otimes \rho'',
\end{align*}
for some matrices $\rho'$ and $\rho''$.

Note that $\mathrm{Tr}( \ket{0} \bra{0} ) = \mathrm{Tr}( \ket{1} \bra{1} ) = 1$
and $\mathrm{Tr}( \ket{0} \bra{1} ) = \mathrm{Tr}( \ket{1} \bra{0} ) = 0$.
By the definition of the partial trace,
\begin{align*}
&\mathrm{Tr_1} ( \rho )\\
&=
\biggl(
     \frac{ \alpha }{\sqrt{2}} \frac{ \overline{\alpha} }{\sqrt{2}}
        \sum_{\bm{a }, \bm{ a'} \in A_0 } \ket{ \bm{a} } \bra{ \bm{a '} }
  +
     \frac{ \alpha }{\sqrt{2}} \frac{ \overline{\beta} }{\sqrt{6}}
        \sum_{\bm{a } \in A_0, \bm{b' } \in B_0 } \ket{ \bm{a} } \bra{ \bm{b'} }\\
 &+
     \frac{ \beta }{\sqrt{6}} \frac{ \overline{\alpha} }{\sqrt{2}}
        \sum_{\bm{b } \in B_0, \bm{ a'} \in A_0 } \ket{ \bm{b} } \bra{ \bm{a '} }
  +
     \frac{ \beta }{\sqrt{6}} \frac{ \overline{\beta} }{\sqrt{6}}
        \sum_{\bm{b }, \bm{b' } \in B_0 } \ket{ \bm{b} } \bra{ \bm{b'} }
   \biggr)\\
&+ \biggl(
     \frac{ \alpha }{\sqrt{2}} \frac{ \overline{\alpha} }{\sqrt{2}}
        \sum_{\bm{a }, \bm{ a'} \in A_1 } \ket{ \bm{a} } \bra{ \bm{a '} }
  +
     \frac{ \alpha }{\sqrt{2}} \frac{ \overline{\beta} }{\sqrt{6}}
        \sum_{\bm{a } \in A_1, \bm{b' } \in B_1 } \ket{ \bm{a} } \bra{ \bm{b'} }\\
&+
     \frac{ \beta }{\sqrt{6}} \frac{ \overline{\alpha} }{\sqrt{2}}
        \sum_{\bm{b } \in B_1, \bm{ a'} \in A_1 } \ket{ \bm{b} } \bra{ \bm{a '} }
  +
     \frac{ \beta }{\sqrt{6}} \frac{ \overline{\beta} }{\sqrt{6}}
        \sum_{\bm{b }, \bm{b' } \in B_1 } \ket{ \bm{b} } \bra{ \bm{b'} }
   \biggr)\\
&=
\frac{1}{2}\biggl(
     \alpha \overline{\alpha} 
        \sum_{\bm{a }, \bm{ a'} \in A_0 } \ket{ \bm{a} } \bra{ \bm{a '} }
  +
     \alpha \frac{ \overline{\beta} }{\sqrt{3}}
        \sum_{\bm{a } \in A_0, \bm{b' } \in B_0 } \ket{ \bm{a} } \bra{ \bm{b'} }\\
 &+
     \frac{ \beta }{\sqrt{3}} \overline{\alpha}
        \sum_{\bm{b } \in B_0, \bm{ a'} \in A_0 } \ket{ \bm{b} } \bra{ \bm{a '} }
  +
     \frac{ \beta }{\sqrt{3}} \frac{ \overline{\beta} }{\sqrt{3}}
        \sum_{\bm{b }, \bm{b' } \in B_0 } \ket{ \bm{b} } \bra{ \bm{b'} }
   \biggr)\\
&+ \frac{1}{2} \biggl(
     \alpha \overline{\alpha} 
        \sum_{\bm{a }, \bm{ a'} \in A_1 } \ket{ \bm{a} } \bra{ \bm{a '} }
  +
     \alpha \frac{ \overline{\beta} }{\sqrt{3}}
        \sum_{\bm{a } \in A_1, \bm{b' } \in B_1 } \ket{ \bm{a} } \bra{ \bm{b'} }\\
&+
     \frac{ \beta }{\sqrt{3}} \overline{\alpha}
        \sum_{\bm{b } \in B_1, \bm{ a'} \in A_1 } \ket{ \bm{b} } \bra{ \bm{a '} }
  +
     \frac{ \beta }{\sqrt{3}} \frac{ \overline{\beta} }{\sqrt{3}}
        \sum_{\bm{b }, \bm{b' } \in B_1 } \ket{ \bm{b} } \bra{ \bm{b'} }
   \biggr)\\
&= \frac{1}{2}| \Phi_0 \rangle \langle \Phi_0 | +  \frac{1}{2} | \Phi_1 \rangle \langle \Phi_1 |.
\end{align*}

By the symmetry of $\ket{\Phi}$,
$$
D_1 (\rho) = D_2( \rho) = D_3 (\rho) = D_4 (\rho)
$$
holds.
\end{proof}

\begin{theorem}
The code $Q_4$ is a $[4,1]$ single deletion error-correcting code
with the encoder $\mathrm{En}_4$ and the decoder $\mathrm{De}_4$.
\end{theorem}
\begin{proof}
Since $\ket{\Phi_0} \in V_0$ and $\ket{\Phi_1} \in V_1$,
by the Step 1,
The obtained quantum state $s$ is $\ket{\Phi_0} \bra{\Phi_0}$
or $\ket{\Phi_1} \bra{\Phi_1}$.

At the Step 2,
if the outcome is $1$,
the state is changed to $\ket{\Phi_0} \bra{\Phi_0}$ by the operation $F$.
Hence after the Step 2, obtained quantum state is $\ket{\Phi_0} \bra{\Phi_0}$.

Since $G( \Phi_0 ) = \alpha \ket{000} + \beta \ket{100}
 = \ket{\phi} \otimes \ket{00}$,
by the Step 3,
the quantum state is changed to
$
\ket{ \phi } \bra{ \phi} \otimes \ket{00} \bra{00}.
$

Hence at the Step 4,
the obtained state is $\ket{\phi} \bra{ \phi }$,
i.e. the original quantum state.
\end{proof}

\section{Encoding Circuit and Decoding Circuit}
\subsection{Encoding Circuit}
\begin{figure}[htbp]
\begin{center}
\includegraphics[width=5.5cm,bb=0 0 145 136]{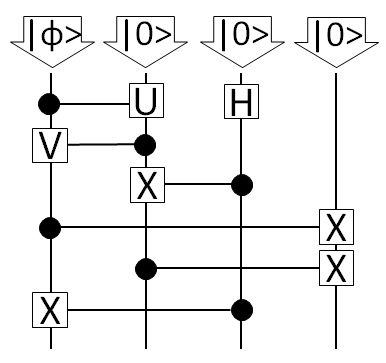} 
\caption{Encoder }
\label{figure:encoder}
\end{center}
\end{figure}

Figure \ref{figure:encoder} shows an example of an encoder for our four qubits code $Q_4$.
The gate $H$ is the Hadamard matrix,
i.e.,
$$
H = 
\left(
\begin{array}{cc}
1/\sqrt{2} & 1/\sqrt{2} \\
1/\sqrt{2} & -1/\sqrt{2}
\end{array}
\right),
$$
and the gate $X$ is the bit-flip matrix,
i.e.,
$$
X := 
\left(
\begin{array}{cc}
0 & 1 \\
1 & 0
\end{array}
\right).
$$
The black dot is the control gate for the connected gate.
For example, the final gate of Figure \ref{figure:encoder}
is the C-NOT gate with the 3rd qubit as the control qubit
and the 1st qubit as the target bit.

The gates $U$ and $V$ are defined as
$$
U := 
\left(
\begin{array}{cc}
1/\sqrt{3} & -\sqrt{2}/\sqrt{3} \\
\sqrt{2}/\sqrt{3} & 1/\sqrt{3}
\end{array}
\right),
V := 
\left(
\begin{array}{cc}
1/\sqrt{2} & 1/\sqrt{2} \\
-1/\sqrt{2} & 1/\sqrt{2}
\end{array}
\right)
$$
respectively.

The input of the circuit is a single qubit $\ket{ \phi } \in \mathbb{C}^2$.
The input position is settled to the left-most position.
For the other three positions,
initialized three qubits $| 000 \rangle$ are input.

By the first two depth operations,
i.e., the controlled $U$ gate, the controlled $V$ gate and the Hadamard gate,
a quantum state $\ket{ 0000 }$ is changed to
$$
\frac{1}{ \sqrt{2}} \left( \ket{ 0000 } + \ket{ 0010 } \right).
$$
and a quantum state $\ket{ 1000 }$ is changed to
$$
\frac{1}{ \sqrt{6}} \left( \ket{ 1000 } + \ket{ 1010 }
+
 \ket{ 0100 } +  \ket{ 0110 }
+
 \ket{ 1100 } +  \ket{ 1110 } \right).
$$

By the remaining four controlled  $X$ gates, i.e. C-NOTs,
a quantum state $\ket{ x_1 x_2 x_3 0 }$ is changed to
$\ket{(x_1 + x_3) (x_2 + x_3) (x_3) (x_1 + x_2 + x_3)}$,
e.g.,
\begin{align*}
\ket{0000} &\rightarrow \ket{0000},\\
\ket{0010} &\rightarrow \ket{1111},\\
\ket{1000} &\rightarrow \ket{1001},\\
\ket{1010} &\rightarrow \ket{0110},\\
\ket{0100} &\rightarrow \ket{0101},\\
\ket{0110} &\rightarrow \ket{1010},\\
\ket{1100} &\rightarrow \ket{1100},\\
\ket{1110} &\rightarrow \ket{0011}.
\end{align*}

Hence $(\alpha \ket{ 0 } + \beta \ket{ 1} ) \otimes \ket{ 000 }$ is encoded to
$$
\alpha \left( \frac{1}{\sqrt{2}} \sum_{ \bm{a} \in A} \ket{ \bm{a} } \right)
+
\beta  \left( \frac{1}{\sqrt{6}} \sum_{ \bm{b} \in B} \ket{ \bm{b} } \right).
$$
\subsection{Decoding Circuits}

\begin{figure}[htbp]
\begin{center}
\includegraphics[width=4cm,bb=0 0 108 105]{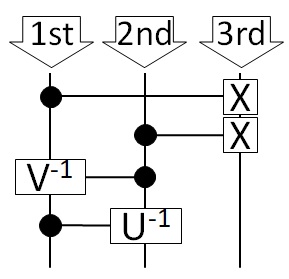} 
\caption{Decoding Circuit after Step 1}
\label{figure:decoder_3}
\end{center}
\end{figure}

Figure \ref{figure:decoder_3} shows a quantum circuit that
changes a pure state $\left( \ket{011} + \ket{101} + \ket{110} \right)/\sqrt{3}$ to a pure state $\ket{100}$
and keeps a pure state $\ket{000}$.
Hence the circuit changes a pure state $\ket{ \Psi_{0} }
= \alpha \sum_{ \bm{a} \in A_0 } \ket{ \bm{a} } + \frac{\beta}{3} \sum_{ \bm{b} \in B_0 } \ket{ \bm{b} }$
to a pure state $\ket{ \phi } = \alpha \ket{0} + \beta \ket{1}$.

Remember that the step 1 of the decoding algorithm
changes a single deleted quantum state $D_i ( \rho )$ to $\ket{\Psi_0}$.
Hence the circuit is available as a decoding circuit after step 1.

Let us provide another quantum circuit that is depicted as Figure \ref{figure:decoder_4}.
The quantum circuit consists of two parts.
The first part has six C-NOT and the last part is the same as the circuit defined by Figure \ref{figure:decoder_3}.
The first part changes a quantum pure state $\ket{x_1 x_2 x_3}$
 to $\ket{x_1 x_2 x_3 0}$ if the Hamming weight of $x_1 x_2 x_3$ is even and
 to $\ket{(x_1 + 1) (x_2 +1)  (x_3 + 1) 1}$ if the weight is odd.
This implies that 
the first part changes
a single deleted quantum state $D_i (\rho)$ to 
$$ \ket{\Psi_0} \bra{\Psi_0} \otimes \left( \begin{array}{cc} 1/2 & 0 \\ 0 & 1/2 \end{array} \right).$$

Since the last part is the same as the circuit corresponding to Figure \ref{figure:decoder_3},
the state is changed to 
$$\ket{ \phi} \bra{\phi} \otimes \ket{0} \bra{0} \otimes \ket{0} \bra{0} 
\otimes \left( \begin{array}{cc} 1/2 & 0 \\ 0 & 1/2 \end{array} \right).$$
In other words, the state at the first position of output is a pure state $\ket{ \phi }$.

\begin{figure}[htbp]
\begin{center}
\includegraphics[width=5.5cm,bb=0 0 144 214]{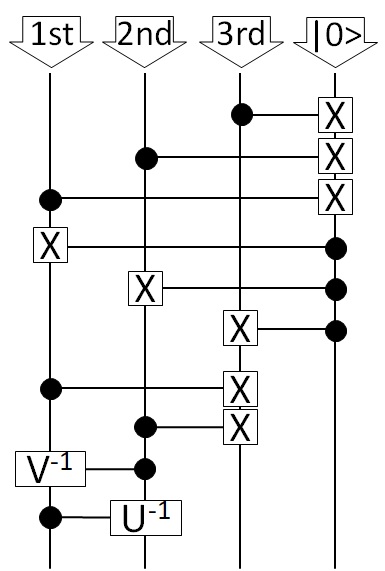} 
\caption{Decoding Circuit without Measurement}
\label{figure:decoder_4}
\end{center}
\end{figure}

\section{Generalization}
\subsection{Number of Information Qubits}
We generalize our $[4, 1]$ single deletion error-correcting code $Q_4$ to
a $[2^{k+2} - 4, k]$
single quantum deletion error-correcting code
for any positive integer $k$.

Let $l$ be a positive integer.
For $0 \le i \le l-1$,
let us define
$$
\mathcal{A}_i := \{ \bm{x} \in \{ 0, 1 \}^{4(l-1)} \mid
 \mathrm{wt}( \bm{x} ) = 2i \text{ or } 4(l-1) - 2i\}
$$
and
$$
\mathrm{En}_{4 (l-1)} \left( \ket{ \phi_l } \right)
:=
\sum_{0 \le i \le l-1} c_i \left( \sum_{\bm{x} \in \mathcal{A}_i } \ket{ \bm{x} } \right),
$$
where $\ket{ \phi_l } := \sum_{0 \le i \le l-1} c_i \ket{i}$ is a quantum pure state of level $l$,
i.e. $\ket{ \phi_l } \in \mathbb{C}^{l}$,
and $\ket{0}, \ket{1}, \dots, \ket{l-1}$ is the standard orthonormal basis of $\mathbb{C}^l$.

We claim that the image of $\mathrm{En}_{4(l-1)}$ is a single quantum deletion error-correcting code
for any $l \ge 2$.
A decoding algorithm $\mathrm{De}_{4(l-1)}$ can be defined similar to $\mathrm{De}_{4}$.
Due to the limit of page numbers, we only explain how to define its measurement $\{ \mathcal{P}_0, \mathcal{P}_1 \}$.

$\mathcal{P}_0$ (resp. $\mathcal{P}_1$) is the projection from $\mathbb{C}^{2 \otimes 4(l-1)}$ to
the linear space $\mathcal{V}_0$ (resp. $\mathcal{V}_1$) spanned by 
\begin{align*}
 \{ \ket{ \bm{x}} \mid & \; \bm{x} \in \{0, 1\}^{4(l-1) - 1}, \\
 & \text{ the Hamming weight of $\bm{x}$ is even (resp. odd) }
  \}.
\end{align*}
The outcome is $i \in \{0, 1 \}$ if the quantum state in $S( \mathbb{C}^{2 \otimes 4(l-1)})$
 is changed into a state in $S( \mathcal{V}_{i} )$.

For the case $l= 2^k$, we obtain a $[2^{k+2} - 4, k]$ single quantum deletion error-correcting code.

\subsection{Permutation of The Received Qubits}

Our code word has symmetry for position permutations.
In other words, any permutation for four qubits does not change the codeword.
Additionally, any received word after any single deletion has also symmetry for position permutation.

This means that even if we permute the three input to the quantum circuits
corresponding to Figure \ref{figure:decoder_3},
we can obtain the original quantum information $\ket{ \phi }$ at the left most position of output.

\section{There is no Quantum Deletion Codes with less than $4$ Qubits}

\begin{lemma}\label{thm:noLeng2}
There is no single deletion error-correcting code of length $2$.
\end{lemma}
\begin{proof}
Let us assume that there exist a code of length $2$.
Then we have an encoder $\mathrm{En}_2 : \mathbb{C}^2 \rightarrow \mathbb{C}^{2 \otimes 2}$
and a decoder $\mathrm{De}_2 : S(\mathbb{C}^2) \rightarrow \mathbb{C}^2$.
Let $\rho \in \mathbb{C}^{2}$ and encode it to $\mathrm{En}_2 ( \rho )$.
Assume that the encoded state is corresponding to two photons $p_1$ and $p_2$.
The quantum states of these photons are $D_1 \circ \mathrm{En}_2 (\rho)$ and $D_2 \circ \mathrm{En}_2 ( \rho )$ respectively.
Perform the decoder $\mathrm{De}_2$ to the photons $p_1$ and $p_2$ simultaneously.
Then the states for $p_1$ and $p_2$ are changed to $\mathrm{De}_2 \circ D_1 \circ \mathrm{En}_2 (\rho) = \rho$ and $\mathrm{De}_2 \circ D_2 \circ \mathrm{En}_2 (\rho) = \rho$
respectively.
It contradicts to non-cloning theorem.
\end{proof}

The following is the remarkable result by Grassl et.al.
\begin{fact}[Theorem 5 \cite{grassl1997codes}]\label{fact:noLeng3}
There is no quantum error-correcting code of length
three that can correct one erasure and encodes one qubit.
\end{fact}

A similar result to deletion error-correcting codes holds.
\begin{lemma}\label{thm:noLeng3}
There is no single deletion error-correcting code of length $3$.
\end{lemma}
\begin{proof}
Assume that there exists a single deletion error-correcting code
of length $3$.
Let us explain that this code can be regarded as
a quantum erasure error-correcting code of length $3$.
If we have a received word with a single quantum erasure error,
we can find the position where the error occurs.
Then delete the state at the position of the received word.
In other words, we can convert an erasure error to 
a deletion error.
By the assumption,
we can correct the error by deletion error-correcting decoder.
This contradicts to Fact \ref{fact:noLeng3}.
\end{proof}

\begin{theorem}
The shortest length of single quantum deletion error-correcting codes is $4$.
\end{theorem}
\section{Conclusion}
This paper provided a single quantum deletion error-correcting code
of optimal length, i.e. $4$.
The construction of this code is far from known quantum error-correcting codes,
e.g.,
CSS codes \cite{calderbank1996good,steane1996multiple},
stabilizer codes \cite{gottesman1997stabilizer},
surface codes \cite{fowler2012surface},
and etc.

\section*{Acknowledgment}
This paper is partially supported by
KAKENHI 18H01435.
The authors thank to Professor Mikio Nakahara,
Professor Akinori Kawachi, and
Professor Akihisa Tomita for valuable discussion.

\bibliography{reference}

\end{document}